\documentclass[10pt,nocopyrightspace]{sigplan-proc-varsize-1in}
\usepackage{hyperref}
\usepackage{cite}
\usepackage{url}

\usepackage{graphicx}
\usepackage{caption}
\usepackage{subfigure}

\usepackage{algorithm}
\usepackage[noend]{algorithmic}
\usepackage{amssymb, amsmath,graphicx,charter, latexsym}
\newtheorem{definition}{Definition}
\newtheorem{lemma}{Lemma}
\newtheorem{theorem}{Theorem}
\newtheorem{example}{Example}

\def\baselinestretch{0.99}
\begin{document}
\title{Capacity and Scheduling of Access Points for Multiple Live Video Streams}
\numberofauthors{2}
\author{
\alignauthor
I-Hong Hou\\
\affaddr{CESG and Department of ECE}\\
\affaddr{ Texas A\&M University}\\
\affaddr{College Station, TX 77843, USA}\\
\email{ihou@tamu.edu}
\alignauthor
Rahul Singh\\
\affaddr{CESG and Department of ECE}\\
\affaddr{ Texas A\&M University}\\
\affaddr{College Station, TX 77843, USA}\\
\email{rsingh1@tamu.edu} 
}
\maketitle
\begin{abstract}
This paper studies the problem of serving multiple live video streams to several different clients from a single access point over unreliable wireless links, which is expected to be major a consumer of future wireless capacity. This problem involves two characteristics. On the streaming side, different video streams may generate variable-bit-rate traffic with different traffic patterns. On the network side, the wireless transmissions are unreliable, and the link qualities differ from client to client. In order to alleviate the above stochastic aspects of both video streams and link unreliability, each client typically buffers incoming packets before playing the video. The quality of the video playback subscribed to by each flow depends, among other factors, on both the delay of packets as well as their throughput.

In this paper we characterize precisely the capacity of the wireless video server in terms of what combination of joint per-packet-delays and throughputs can be supported for the set of flows, as a function of the buffering delay introduced at the server.  

We also address how to schedule packets at the access point to satisfy the joint per-packet-delay-throughput performance measure. We test the designed policy on the traces of three movies. From our tests, it appears to outperform other policies by a large margin.

\end{abstract}

\section{Introduction}  \label{section:introduction}

Multimedia applications for video streaming have been predicted to become a dominant portion of future wireless traffic \cite{cisco}. In order to provide smooth playback to end users with multimedia applications, their packets need to be delivered in a timely manner and with a sufficient throughput. Providing the per-packet delay guarantees is particularly challenging for two reasons. First, on the application side, the video streams generate variable-bit-rate (VBR) traffic that congests the network when multiple streams generate a burst of packets at the same time. Second, on the network side, wireless transmissions are inherently unreliable due to shadowing, fading, and interference.

In order to mitigate the effect of the uncertainties in traffic bit rate as well as channel reliability, packets are typically buffered at the receivers. Each receiver waits a specified amount of time buffering incoming packets before it starts playing the video. On the positive side, this approach has the benefit of greatly reducing the impact of network congestion and channel unreliability. However, buffering also increases delay as the receiver cannot start playing the video immediately. For live video streaming, such as sports events and instant news, it is important to provide smooth playback while waiting only a reasonably small amount of time before playing the video.

In this paper, we analyze a model that addresses the three characteristics of stochastic packet arrivals, unreliable wireless channel, and buffering delay. Since the performance of the flows depends on the buffering policy, we characterize the precise combination of per-packet-delays and throughputs of the several video flows that the access point can achieve, for each value of the buffering delay deliberately introduced at the access point. We also determine a simple online scheduling policy, the Earliest Positive-Debt Deadline First (EPDF) policy, for serving the live video streams that achieves the delay-throughput capacity. This policy does not require knowledge of the exact traffic pattern of each video stream, enhancing its suitability for implementation.

We test the proposed EPDF policy on three movie traces. We compare its performance with three other well known policies in ns-2. The EPDF policy outperforms the other policies by a large margin on these three video traces. We further investigate a range of practical issues, including the tradeoff between per-packet delays and achievable throughputs. We also propose a modification to enhance short-term performance, and investigate the tradeoff between long-term and short-term performance guarantees. While theoretical studies show that long-term performance cannot be guaranteed when the EPDF policy is modified, our simulation results suggest that, for practical scenarios where packets are not adversarily generated, improving short-term performance only results in negligible impact on long-term performance.

The rest of the paper is organized as follows: Section \ref{section:related} summarizes existing work. Section \ref{section:model} introduces the analytical model for live video streaming over wireless links. Section \ref{section:feasibility} characterizes the per-packet-delay versus throughput for the buffering delay introduced. Section \ref{section:scheduling} proposes an online scheduling policy and proves that it is capacity achieving Section \ref{section:simulations} provides the simulation results based on real video traces. Finally, Section \ref{section:conclusions} concludes the paper.

\section{Related Work}	\label{section:related}

There have been several studies on providing delay guarantees for flows that require delay guarantees over wireless networks. Liu, Wang, and Giannakis \cite{QL06} have proposed a scheduler that updates link priorities dynamically to provide QoS. Dua and Bambos \cite{AD07} have proposed a heuristic that jointly considers wireless channel quality and packet deadlines. Raghunathan et al \cite{VR08} and Shakkottai and Srikant \cite{SS02} have proposed scheduling policies that minimize the number of packets that exceed their delay bounds. These studies cannot provide guarantees on the throughput of packets delivered on-time for each flow.

Hou el al \cite{IH09} have proposed a model that jointly considers the per-packet delay bound and the per-flow throughput requirement. This model has been extended to consider variable-bit-rate traffic \cite{IH09Hoc}, fading wireless channels \cite{IH10}\cite{JJ11}, the mixture of real-time and non-real-time traffic \cite{JJ10}, and multi-hop wireless transmissions \cite{RL11}. However, these studies assume that all flows are synchronized and generate packets at the same time, and the results depend critically on this assumption. Moreover, they assume that all flows start playback immediately without buffering any packets, which is a critical feature of the playback process. Dutta el al \cite{PD12} have studied serving video streams when receivers may buffer packets before playing them. They assume, however, that all packets are available at the server when the system starts, which is applicable to on-demand videos, but not to live videos.

Stockhammer, Jenkac, and Kuhn \cite{TS04}, and Liang and Liang\cite{GL07}\cite{GL08} have studied the influence of buffer sizes on the performance of real-time flows. Li, Claypool, and Kinicki \cite{ML09} have studied the problem of choosing the optimal buffer size for wireless multimedia streaming. However, their studies focus only on the scenario where there is only one flow in the system, and ignore the critical feature of contention for the medium.

\section{System Model}  \label{section:model}
\begin{figure}[t]
\subfigure[Network topology]{\label{fig:model:system:topology}
\includegraphics[width=3in]{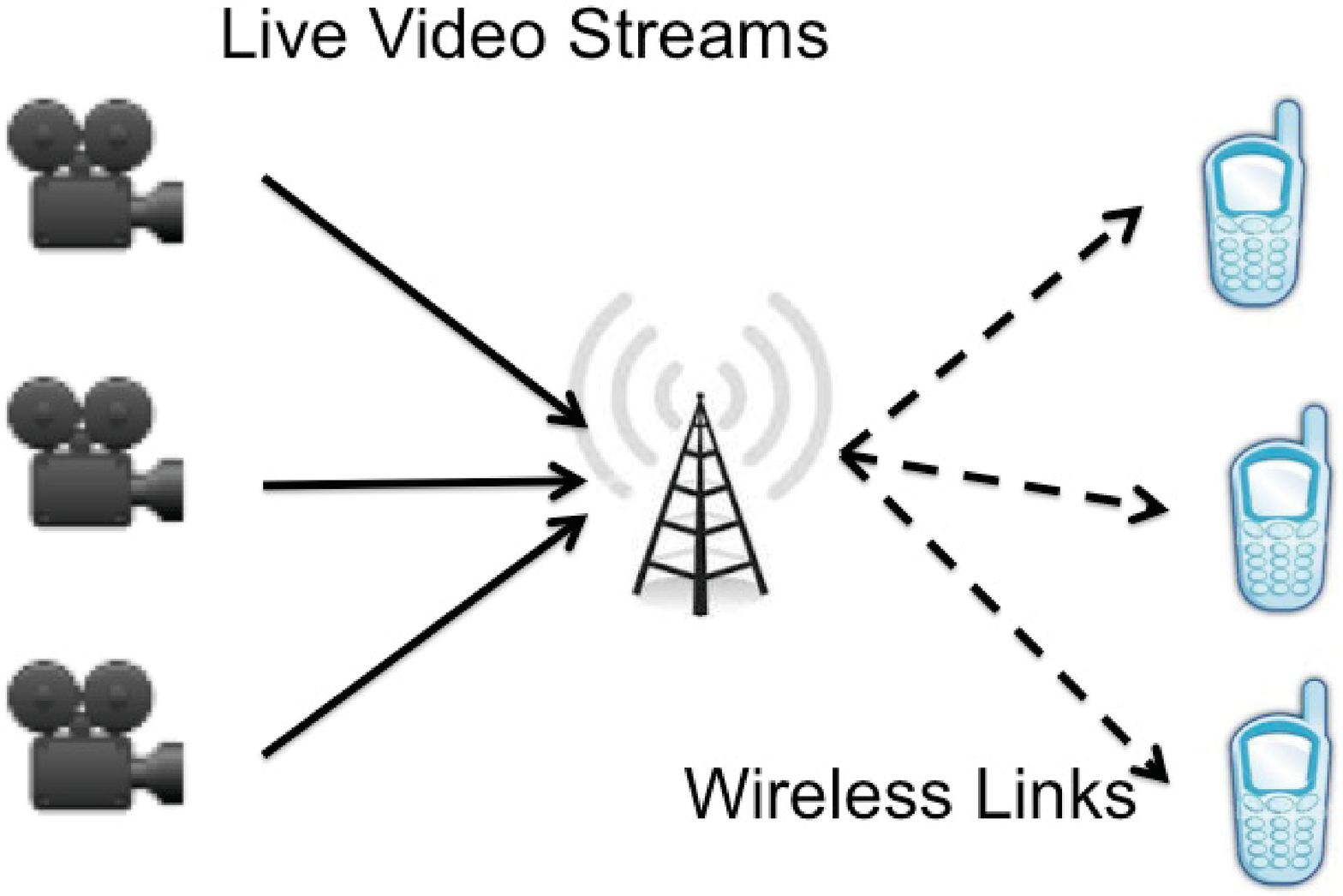}}
\subfigure[Packet arrivals and deadlines of a stream, with sequence numbers indicated.]{\label{fig:model:system:arrival}
\includegraphics[width=3in]{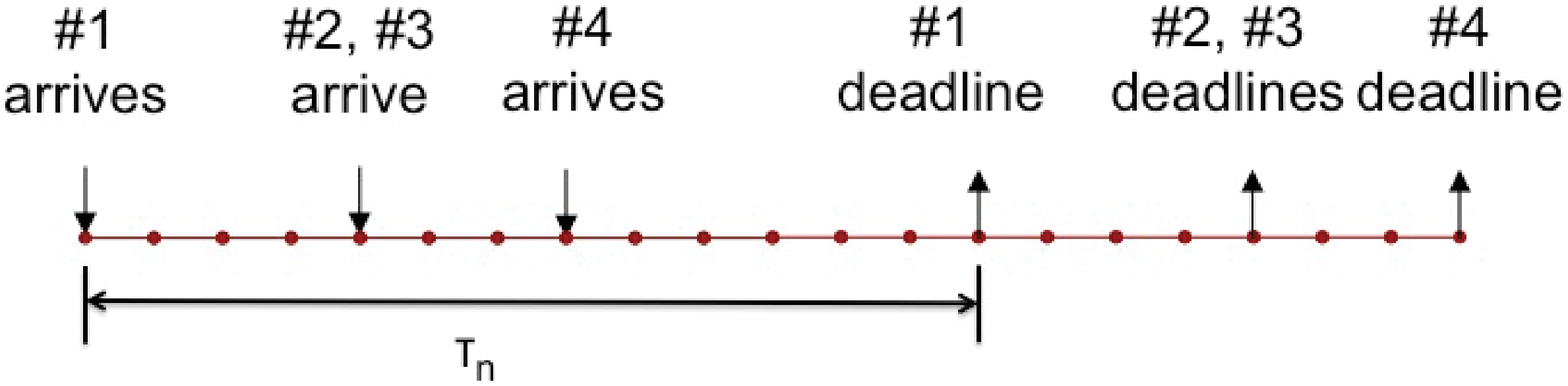}}
\caption{System for wireless live video streaming}\label{fig:model:system}
\end{figure}

Consider a system with $N$ wireless clients, $\{1,2,\dots,N\}$, and one access point (AP). Each client subscribes to a live video stream through the AP. Fig. \ref{fig:model:system:topology} illustrates such a system. 

Video streams generate variable-bit-rate traffic, and the exact number of packets generated in a time slot is influenced by not only the video encoding mechanism, but also the context of the current frame. We suppose that the number of packets generated by each flow in a time slot is described by some irreducible finite-state Markov chain, and is bounded. However, we do not assume that either the AP or the clients know this Markov chain model. We only suppose that each client $n$ knows the long-term average number of packets generated by its subscribed video stream, which we denote by $r_n$.  In each time slot the access point can transmit one such packet, including all overhead.

To smooth the bursty arrival process as well as channel unreliability, each client buffers incoming packets before it commences playing the video. To be more specific, a video frame that is generated at time slot $t$ for client $n$ will be played by the client at time slot $t+\tau_n$, with the AP informed of the value of $\tau_n$. This is equivalent to specifying a hard per-packet delay bound of $\tau_n$ time slots for the video stream of client $n$. Fig. \ref{fig:model:system:arrival} illustrates an example of packet arrivals and their respective deadlines. If a packet cannot be delivered within its delay bound, it is dropped from the system. Packet drops result in glitches in the live videos, and hence need to be limited.

We model the channel reliability between the AP and client $n$ by a probability $p_n$. When the AP schedules a transmission for client $n$, the transmission is successful with probability $p_n$. The AP obtains immediate feedback on whether a transmission is successfully received by ACKs. A transmission is considered to be successful if both the packet and ACK are correctly received. When a transmission fails, the AP may retransmit the packet, if the packet has not already expired, i.e., exceeded its delay bound, or it may schedule a different packet in the next time slot.

For enforcing a minimum quality for each video, each client $n$ requires that its throughput is at least $q_n$ packets per time slot; that is, it requires that $\liminf_{T\rightarrow\infty}\frac{\sum_{t=1}^Te_n(t)}{T}\geq q_n$, where $e_n(t)$ is the indicator function that a packet for client $n$ is delivered in time slot $t$.

We study two important problems for serving live video streams over wireless links. One is to characterize the capacity region of achievable vectors of throughput requirements $[q_n]$, given the traffic pattern, the delay bounds, and the channel reliabilities. We will say that a vector of throughputs $[q_n]$ is strictly in the capacity region if there is a scheduling policy for the access point for which each client $n$ obtains a throughput $q_n+\alpha$ for some $\alpha >0$. The other problem addressed is to design a simple online scheduling policy that can support all vectors of throughputs that are strictly in the capacity region.

\section{Deriving the Capacity Region}	\label{section:feasibility}

We now address the problem of deriving the capacity region of a system. Since packet generations of each live video stream are assumed to be modeled as an irreducible finite-state Markov chain, we have the following:

\begin{lemma}	\label{lemma:feasibility:Markov}
If a vector of throughputs is in the capacity region, it can be supported by a randomized stationary scheduling policy that makes scheduling decisions solely based on the set of packets then available for transmission and their respective deadlines.
\end{lemma}

Let $\Psi_{n,\tau}(t)$ be the number of packets for client $n$ that will expire in $\tau$ time slots at time $t$, i.e., the number of packets for client $n$ that have deadlines at $\tau+t$. We can represent the state of the system at time $t$ by the matrix containing $\Psi_{n,\tau}(t)$ for all $n$ and $1\leq \tau\leq \max_n\{\tau_n\}$, as well as the state of the traffic pattern of each video. Since the number of packets generated by a stream in a time slot is assumed to be bounded, $\Psi_{n,\tau}(t)$ is also bounded, and hence the number of possible system states is finite.

Next, we consider the transition of the system state. The states of the traffic patterns of videos evolve without being influenced by the scheduling policy of the AP, and only the matrix containing $\Psi_{n,\tau}(t)$ is influenced by the AP. Suppose that the AP schedules a packet of client $m$ that will expire in $\hat{\tau}_m$ time slots at time $t$. This transmission will be successful with probability $p_m$, in which case the packet leaves the system. All other packets continue to remain in the system at the next time slot unless they have expired, but their deadlines will be one time slot closer. Hence, $\Psi_{m,\hat{\tau}_m-1}(t+1) = \Psi_{m,\hat{\tau}_m}(t)-1$ with probability $p_m$, and $\Psi_{m,\hat{\tau}_m-1}(t+1) = \Psi_{m,\hat{\tau}_m}(t)$ with probability $1-p_m$ if $\hat{\tau}_m-1\geq 0$. For all $n\neq m$ or $\tau\neq \hat{\tau}_m$, $\Psi_{n,\tau-1}(t+1)=\Psi_{n,\tau}(t)$ if $0\leq \tau - 1\leq \tau_n-1$. Finally, $\Psi_{n,\tau_n}(t+1)$ is the number of packets generated for client $n$ at time $t+1$. In each time slot $t$, the AP obtains the vector of rewards $[e_n(t)]$, where $e_n(t)$ is the indicator function that a packet for client $n$ is delivered in time slot $t$. Thus, if the AP schedules $m$ in time slot $t$, we have $e_m(t)=1$ with probability $p_m$, and $e_n(t)=0$ for all $n\neq m$.

With the system being modeled as a controlled finite-state Markov chain, a stationary randomized policy $\eta$ obtains a vector of long-term average rewards equal to the vector $[\lim_{T\rightarrow\infty}\frac{\sum_{t=1}^Te_n(t)}{T}]$. We denote the throughput  $\lim_{T\rightarrow\infty}\frac{\sum_{t=1}^Te_n(t)}{T}$ of client $n$ under $\eta$ by $\hat{q}_n(\eta)$. We can then consider the infinite horizon optimization problem of $\min_\eta \sum_n(q_n-\hat{q}_n(\eta))^+$, where $x^+:=\max\{x,0\}$. The vector of throughputs is within the capacity region if and only if $\min_\eta \sum_n(q_n-\hat{q}_n(\eta))^+=0$. The policy $\arg\min_\eta \sum_n(q_n-\hat{q}_n(\eta))^+$ then supports the vector of throughputs.

\section{The Homogeneous Case: Characterization of Capacity and Optimality of EDF}

Although the above procedure characterizes the capacity region, it involves solving the infinite horizon optimization problem of a Markov chain that has at least $B^{\sum_n\tau_n}$ states, where $B$ is the maximum number of packets that a stream can generate in a time slot. Hence, it is usually intractable. Before proceeding further, we consider a special homogeneous case that turns out to have an explicit characterization. We assume that time slots are grouped into \emph{intervals}, where the $k$-th interval consists of $T$ consecutive time slots in $(kT, (k+1)T]$. At the beginning of each interval, that is, at time slots $1,T+1,2T+1,\dots$, each stream generates a packet. We will assume the homogeneous case where all streams have the same delay bounds, which is an integral multiple of the interval length, i.e., $\tau_n\equiv KT$, for some positive integer $K$. The special case $K=1$ corresponds to the scenario where each client plays its video stream immediately without buffering, and its capacity region has been characterized in \cite{IH09}. For this homogeneous case, we also assume that all the clients have the same channel reliability and the same throughput requirement, i.e., $p_{n}\equiv p$, and $q_{n}\equiv q$, for all $n$. We will determine the maximum $q$ such that the vector $[q_n|q_n\equiv q]$ is within the capacity region. The above assumptions are made only for an explicit characterization of the capacity region in this section, and are not needed for other sections.

Denote by $w_n :=\frac{q_n}{p_n}$ the implied workload of client $n$. Then, as noted in \cite{IH09}, the throughput of client $n$ is at least $q_{n}$ packets per time slot if and only if the AP, on average, schedules client $n$ for $w_{n}$ times per time slot. Assume that there are $\zeta$ available packets in an interval. The total number of time slots needed to deliver all $\zeta$ packets is the sum of $\zeta$ independent geometric random variables, each with mean $1/p$, and can be represented by $\Gamma(\zeta,p)$. In the event that $\Gamma(\zeta;p)<T$, there are $T-\Gamma(\zeta,p)$ time slots that cannot be utilized in the interval, and are hence forced to idle.

The average number of idle time slots may be different from policy to policy, due to the buffering employed, which strongly distinguishes the problem considered from the model analyzed in \cite{IH09}. Let $I^{\eta}$ denote the average number of idle time slots under the policy $\eta$ per interval. It is obvious that the vector $[q_n|q_n\equiv q]$ is within the capacity region only if $\sum_nw_n=N\frac{q}{p}\leq 1-\min_{\eta}I^{\eta}/T$, that is, the sum of work loads for all clients cannot exceed $1-\min_{\eta}I^{\eta}/T$, which is the maximum possible fraction of non-idle time slots. We will next discuss how to derive the value $\min_{\eta}I^{\eta}$, and show that this condition is also sufficient.

Let us begin by considering the following \emph{Earliest Deadline First} (EDF) policy. At the beginning of each interval, all the packets in the system are sorted in ascending order of the number of time slots remaining before the packet expires, with ties broken uniformly at random. The EDF policy schedules packets according to this ordering, where a packet is scheduled only after all packets that are higher in the ordering. This policy has been well studied for real-time systems \cite{CLL73}, and several optimality properties have been established. However, unlike in the real-time systems, where it is assumed that the maximum number of time slots needed to deliver a packet is known and bounded (and in fact is considered to be known when obtaining sharp results), in our context, the number of time slots needed to deliver a packet is a random, specifically geometric, random variable with mean $1/p$, and is not bounded. Hence, results in real-time system literature cannot be applied to live video streaming over wireless links.

We now show how to determine the performance of the EDF policy. At the beginning of an interval, let $\zeta(k)$ be the number of packets that will expire in $k$ intervals, i.e. in $k T$ time slots. We represent the state of the system by $\zeta:=[\zeta(1), \zeta(2), \dots, \zeta(K)]$. Suppose that in this interval, a total of $j$ transmissions are successful, which occurs with probability ${T\choose j}p^j(1-p)^{T-j}$. By the design of the EDF policy, packets that expire in $k$ intervals will be scheduled only when all packets that have closer deadlines are delivered, i.e., only when $j\geq\sum_{i=1}^{k-1}\zeta(i)$. At the beginning of the next interval, the number of packets that will expire in $k$ intervals is then $\{\zeta(k+1)-[j-\sum_{i=1}^k\zeta(i)]^+\}^+$, for all $k< K$. The number of packets that expire in $K$ intervals, which is the number of the newly generated packets at the beginning for the interval, is $N$, and so we have $\zeta(K)=N$. When the state of the system is $\zeta$ in an interval, the expected number of idle time slots in the interval is $E\{(T-\Gamma(\sum_{i=1}^{K}\zeta(i);p))^+\}$. Therefore, we can model the system state $\zeta$ as a Markov chain. As $\zeta(k)\leq N$ for all $k$, the number of states in this Markov chain is finite. Using standard techniques for Markov chains, we can solve for its steady state distribution, and thus derive the average number of idle time slots under the EDF policy, which we denote by $I^{E}$. By symmetry, all clients have the same throughput under the EDF policy, and hence the throughput of a client can be written as $\hat{q}^{E}:=\frac{p(T-I^{E})}{NT}$.

Now, we show that $I^{E}$ is $\min_\eta I^{\eta}$.
\begin{theorem}	\label{theorem:feasibility:EDF}
The EDF policy minimizes $I^{\eta}$.
\end{theorem}
\begin{proof}
It suffices to show that $I^{E}\leq I^{\eta}$ for any work-conserving policy $\eta$, i.e., a policy that does not idle whenever there is a packet available for transmission. When the state of an interval is $[\zeta(1),\zeta(2),\dots]$, the expected number of idle time slots in the interval is $E\{(T-\Gamma(\sum_{i=1}^{K}\zeta(i);p))^+\}$, which is a decreasing function of $\sum_{i=1}^{K}\zeta(i)$. Hence we can prove that $I^{E}\leq I^{\eta}$ by showing that the EDF policy maximizes $\sum_{i=1}^{K}\zeta(i)$ in every interval, over all work-conserving policies.

To be more specific, we prove the following claim: Let $j(l)$ be the number of successful transmissions in the $l$-th interval. Denote the state of the system at the beginning of the $l$-th interval by $\zeta_l=[\zeta_l(1),\zeta_l(2),\dots]$. Fix $j(1),j(2),\dots$. Then, at the beginning of each interval $l$, the value of $\sum_{i=k}^K\zeta_l(i)$ under the EDF policy is larger than that under $\eta$, for all $k$.

Let the state of the system at the beginning of the $l$-th interval under the EDF policy, and $\eta$, be $\zeta_l^E=[\zeta_l^E(1),\zeta_l^E(2),\dots]$, and $\zeta_l^\eta=[\zeta_l^\eta(1),\zeta_l^\eta(2),\dots]$, respectively. We prove the claim by induction on $l$. When $l=1$, i.e., at the beginning of the first interval, both policies have not scheduled any transmissions yet. Hence, the claim holds.

Assume that the claim holds for the $l$-th interval. We now show that it also holds for the $(l+1)$-th interval. For each $k$, we have
\begin{align*}
&\sum_{i=k}^K\zeta_{l+1}^\eta(i)\\
\leq& N+\{\sum_{i=k}^{K-1}\zeta_{l}^\eta(i+1)-[j(l)-\sum_{i=1}^{k}\zeta_{l}^\eta(i)]^+\}^+\\
\leq& N+\{\sum_{i=k}^{K-1}\zeta_{l}^E(i+1)-[j(l)-\sum_{i=1}^{k}\zeta_{l}^E(i)]^+\}^+\\
&\mbox{(by induction hypothesis)}\\
=&\sum_{i=k}^K\zeta_{l+1}^E(i).
\end{align*}
Hence, by induction, we have established that $\sum_{i=k}^K\zeta_{l}^\eta(i)\leq \sum_{i=k}^K\zeta_{l}^E(i)$, for all $k$ and $l$.
\end{proof}

\begin{theorem}\label{t1}
For the homogeneous case, the maximum $q$ such that vector of throughputs $[q_n|q_n\equiv q]$ is within the capacity region is $q=\frac{p(T-I^E)}{NT}$.
\end{theorem}
\begin{proof}
We have shown that $I^E=\min_{\eta}I^\eta$. Hence, it is necessary that $Nq/p\leq 1 - I^E/T$ for the vector of throughputs $[q_n|q_n\equiv q]$ to be in the capacity region. On the other hand, since the EDF policy achieves a throughput of $\hat{q}^E=\frac{p(T-I^{E})}{NT}$, $\hat{q}^E$ is indeed the maximum $q$ such that $[q_n|q_n\equiv q]$ is within the capacity region.
\end{proof}

%Although Theorem \ref{t1} characterizes the feasibility conditions, it can only be used for the special case discussed in this section. In particular, it does not apply to systems where different streams are not synchronized, and different clients may require different timely-throughputs and have different channel reliabilities. The reason is that, for the general case, even though the EDF policy still minimizes $I^{\eta}$, which can be proved by a similar argument as the proof of Theorem \ref{theorem:feasibility:EDF}, it may be unfair to some clients, and hence is no longer feasibility-optimal. In Section \ref{section:simulations}, we will further demonstrate that the EDF policy can result in poor performance by simulations. 

\section{The General Case: Earliest Positive-Debt Deadline Policy and Its Optimality}   \label{section:scheduling}

In this section we return to the non-homogeneous case where the $p_n$'s and $q_n$'s are not the same for all clients.  
We will establish the optimality of a simple online scheduling policy called the Earliest Positive-Debt Deadline First (EPDF). We will show that it
can support every vector of throughputs that
is strictly in the capacity region.  

The policy is based on the concept of \emph{truncated time debt}, which differs from the debt in \cite{IH09}. The choice of $M$ below is related to the usage of a multistep negative drift of the Lyapunov function.

\begin{definition}  \label{definition:scheduling:debt}
The \emph{truncated time debt} of client $n$ at time $t$, denoted by $d_n(t)$, is defined recursively by:
\begin{align*}
d_n(t+1) =& [d_n(t)+Mw_n1( t+1\equiv 1\pmod{M})\\
&-1(\mbox{the AP schedules $n$ in time t+1})]^+,
\end{align*}
where $1(\cdot)$ is the indicator function.
\end{definition}

Note that it increases by $Mw_n$ every $M$ time slots, where $M$ is an adjustable parameter, and decreases by 1 in each time slot that the AP schedules client $n$ for transmission. Finally, we only retain the positive part. We will hereafter say that the $M$ time slots form a \emph{frame}. The choice of $M$ will be discussed in the sequel. It will be the single parameter that will be tuned by the policy to adapt to all the above uncertainties.

We now present the scheduling policy, which we call the \emph{Earliest Positive-Debt Deadline First} (EPDF) policy: In each time slot, the AP schedules the packet with the earliest deadline from those whose associated clients have strictly positive truncated time debts, that is, $d_n(t)>0$. Ties are broken arbitrarily. If the associated client of every packet in the system has $d_n(t)=0$, the AP schedules the packet with the earliest deadline. The AP only idles in a time slot when there are no packets to be transmitted. We note that this policy differs from the EDF policy in that it restricts the AP from serving clients with zero truncated time debts, so as to prevent the AP from providing too much service to some clients while starving others.

The EPDF policy only needs to be tuned with an appropriate choice of $M$. It does not need any explicit information on the traffic pattern of each client besides the knowledge of the implied workload and the choice of $M$. It is a simple online scheduling policy that is readily implementable. We now show that this policy supports every vector of throughputs that is strictly in the capacity region through an appropriate choice of $M$. 

Our proof employs the Lyapunov function $L(k):=\sum_{n=1}^Nd_n(kM)$. In our context the Foster-Lyapunov Theorem is used as follows:

\begin{theorem}\label{theorem:scheduling:lyapunov}
If, under some scheduling policy $\eta$, there exists a positive number $\delta$, and a finite subset $\mathcal{D}_0$ of $\mathbb{R}^N$ such that:
\begin{align}
E\{L(k+1)-L(k)|[d_n(kM)]\} \leq -\delta,&\mbox{ if $[d_n(kM)]\notin \mathcal{D}_0$},\\
E\{L(k+1)-L(k)|[d_n(kM)]\}<\infty, &\mbox{ if $[d_n(kM)]\in \mathcal{D}_0$},
\end{align}
then the throughput of each client $n$ is at least $q_n$.
\end{theorem}

\begin{theorem} \label{theorem:scheduling:optimal}
The EPDF policy supports every vector $[q_n]$ that is strictly in the capacity region, with properly chosen $M$.
\end{theorem}
\begin{proof}
Consider a vector of throughputs $[q_n]$ that is strictly in the capacity region. There exists some $\alpha >0$ such that the vector $[q_n+\alpha]$ is still in the capacity region. Therefore, we can assume that there exists a scheduling policy $\eta$, under which the throughput of each client $n$ is at least $q_n+\alpha$.

Our proof consists of two parts. First, we show that the condition in Theorem \ref{theorem:scheduling:lyapunov} is satisfied under $\eta$. We then show the value of $E\{L(k+1)-L(k)\}$ under the EPDF policy is smaller than that under $\eta$. Hence, the condition in Theorem \ref{theorem:scheduling:lyapunov} is also satisfied under the EPDF, and the throughput of each client $n$ is at least $q_n$ under the EPDF policy.

Since the throughput of each client $n$ is at least $q_n+\alpha$ under $\eta$, we have, under $\eta$, the long-term average number of time slots that the AP schedules client $n$ is at least $(q_n+\alpha)/p_n=w_n+\frac{\alpha}{p_n}$ per time slot. Let $\alpha^*:=\min_{n\in\{1,2,\dots,N\}}\frac{\alpha}{p_n}$. Let $u_n(t)$ be the indicator function that the AP schedules client $n$ in time slot $t$. For large enough $M$, we have
\begin{equation}    \label{equation:schdueling:proof1}
E\{\sum_{t=kM+1}^{(k+1)M}u_n(t)\}\geq Mw_n+M\alpha^*,
\end{equation}
and
\begin{equation}    \label{equation:schdueling:proof2}
Prob\{\sum_{t=kM+1}^{(k+1)M}u_n(t)\geq Mw_n\}>1-\frac{\alpha^*}{2\sum_{n=1}^Nw_n},
\end{equation}
for all $n$ and $k$, regardless of events prior to time slot $kM$. We choose one such $M$ so that $Mw_n$ is an integer for each $n$.

We consider the Lyanpunov drift, $E\{L(k+1)-L(k)\}$, under $\eta$. By definition, $d_n((k+1)M)=[d_n(kM)+Mw_n-\sum_{t=kM+1}^{(k+1)M}u_n(t)]^+$. Since $\sum_{t=kM+1}^{(k+1)M}u_n(t)\leq M$, we have
\begin{equation}
d_n((k+1)M)-d_n(kM)=Mw_n-\sum_{t=kM+1}^{(k+1)M}u_n(t),
\end{equation}
if $d_n(kM)>M$,and
\begin{equation}
d_n((k+1)M)-d_n(kM)\leq [Mw_n-\sum_{t=kM+1}^{(k+1)M}u_n(t)]^+,
\end{equation}
otherwise. If $d_n(kM)\leq M$, for all $n$, we have, under $\eta$,
\begin{align*}
&E\{L(k+1)-L(k)\}\\
=&\sum_{n=1}^NE\{d_n((k+1)M)-d_n(kM)\}\\
\leq& \sum_{n=1}^NMW_n<\infty.
\end{align*}
On the other hand, if there exists a client $n_0$ such that $d_{n_0}(kM)>M$, we have, under $\eta$,
\begin{align*}
&E\{L(k+1)-L(k)\}\\
=&\sum_{n=1}^NE\{d_n((k+1)M)-d_n(kM)\}\\
\leq& [Mw_{n_0}-E\{\sum_{t=kM+1}^{(k+1)M}u_{n_0}(t)\}]\\
&+\sum_{n\neq n_0}E\{[Mw_n-\sum_{t=kM+1}^{(k+1)M}u_n(t)]^+\}\\
\leq& -M\alpha^*+\sum_{n\neq n_0}Mw_nProb\{\sum_{t=kM+1}^{(k+1)M}u_n(t)<Mw_n\}\\
\leq& -M\alpha^*+\sum_{n\neq n_0}Mw_n\frac{\alpha^*}{2\sum_{n=1}^Nw_n}\hspace{10pt}\mbox{(by (\ref{equation:schdueling:proof2}))}\\
\leq& -\frac{M\alpha^*}{2}.
\end{align*}
Thus, by letting $\delta=\frac{M\alpha^*}{2}$, and $\mathcal{D}_0$ be the set of states where $d_n(t)\leq kM$, for all $n$, $\eta$ satisfies the condition in Theorem \ref{theorem:scheduling:lyapunov}.

Next, we show that the value of $E\{L(k+1)-L(k)\}$ under the EPDF policy is smaller than that under $\eta$. We number the packets that can be transmitted some time between time slot $kM+1$ and time slot $(k+1)M$ by $i=1,2,3,\dots$. Packet $i$ may either be one that is generated during time slots $(kM, (k+1)M]$, or be one that is generated before time slot $kM+1$ and is not delivered by time slot $kM+1$. We denote by $n(i)$ the client associated with packet $i$, and by $t(i)$ the time slot that packet $i$ is generated. The deadline of packet $i$ can then be expressed as $t(i)+\tau_{n(i)}-1$. Finally, we let $\gamma(i)$ be the number of transmissions that need to be scheduled for packet $i$ during time slots $(kM, (k+1)M]$ before it can be successfully delivered. Therefore, the packet $i$ can be scheduled at most $\gamma(i)$ times. Since the channel reliability for client $n(i)$ is $p_{n(i)}$, $\gamma(i)$ is a geometric random variable with mean $1/p_{n(i)}$. We will show that the value of $L(k+1)-L(k)$ under the EPDF policy is no larger than that under $\eta$, for all $\gamma(1),\gamma(2),\dots$.

Let $\hat{u}(t)$ be the indicator function that the packet scheduled at time slot $t$, say, packet $i$, has $d_{n(i)}(t)>0$. Since we choose $M$ so that $Mw_n$ is an integer for each $n$, $d_n(t)$ is an integer for each $n$ and $t$. Hence, $L(k+1)-L(k)=\sum_{n=1}^NMw_n-\sum_{t=kM+1}^{(k+1)M}\hat{u}(t)$. Given $\gamma(1),\gamma(2),\dots$, let $\Delta^\eta$ and $\Delta^E$ be the values of $\sum_{t=kM+1}^{(k+1)M}\hat{u}(t)$ under $\eta$ and the EPDF policy, respectively. We can show that the value of $L(k+1)-L(k)$ under the EPDF policy is no larger than that under $\eta$ by showing $\Delta^E\geq \Delta^\eta$.

We modify the schedule under $\eta$ time slot by time slot from time slot $kM+1$ to time slot $(k+1)M$, so that it eventually becomes the same as that under EPDF. We show that in each step of the modification, the value of $\Delta^\eta$ does not decrease. Assume that we have modified all time slots within $(kM,t)$ so that the two policies have the same schedule. We now consider time slot $t$. If $\eta$ and EPDF schedules the same packet in time slot $t$, we do not need to make any modifications. Suppose $\eta$ schedules a packet $i^\eta$ and EPDF schedules another packet $i^E$ at time slot $t$, where we set $i^\eta=0$, or $i^E=0$, if $\eta$, or EPDF, does not schedule any packet and idles at time slot $t$, respectively. We modify the schedule of $\eta$ according to the following different cases. First, consider the case when $d_{n(i^\eta)}(t)>0$ and $d_{n(i^E)}(t)>0$, and therefore $\hat{u}(t)=1$ under either policy. If $\eta$ does not schedule packet $i^E$ during time slots $[t,(k+1)M]$, we can simply modify $\eta$ so that it schedules packet $i^E$ at time slot $t$ without changing its schedule in the following time slots. On the other hand, if $\eta$ schedules packet $i^E$ at some time slots after $t$, we pick one of these time slots, denoted by $t'$, and modify the schedule of $\eta$ so that it schedules $i^E$ at time slot $t$, and schedules $t^\eta$ at time slot $t'$. By the design of EPDF, the deadline of $i^E$ must be no larger than that of $i^\eta$. Therefore, $i^\eta$ can be scheduled at time slot $t'$ without violating its deadline constraint. Under this case, the value of $\Delta^\eta$ is not influenced by the modification. Next, consider the cases when $d_{n(i^\eta)}(t)=0$ and $d_{n(i^E)}(t)>0$, and when $d_{n(i^\eta)}(t)=0$ and $d_{n(i^E)}(t)=0$. If $\eta$ does not schedule packet $i^E$ during time slots $[t,(k+1)M]$, we can simply modify $\eta$ so that it schedules packet $i^E$ at time slot $t$ without changing its schedule in the following time slots. On the other hand, if $\eta$ schedules packet $i^E$ at some time slots after $t$, we pick one of these time slots, denoted by $t'$, and modify the schedule of $\eta$ so that it schedules $i^E$ at time slot $t$, and idles at time slot $t'$. Since, under $\eta$, $\hat{u}(t)=\hat{u}(t')=0$ before the modification, this modification does not decrease the value of $\Delta^\eta$. Finally, we note that it is impossible that $d_{n(i^\eta)}(t)>0$ and $d_{n(i^E)}(t)=0$, as the EPDF policy always schedules a packet whose client has strictly positive truncated time debt as long as there is one. In summary, we can use the above procedure to modify the schedule under $\eta$ so that it eventually becomes the same as that under EPDF without decreasing the value of $\Delta^\eta$ in each step of the modification. Hence, we have shown that $\Delta^\eta\leq\Delta^E$.

We have shown that the value of $E\{L(k+1)-L(k)\}$ under EPDF is smaller than that under $\eta$. Therefore, EPDF also satisfies the condition in Theorem \ref{theorem:scheduling:lyapunov}, as long as the vector $[q_n]$ is strictly in the capacity region. EPDF thus supports every vector of throughputs that is strictly in the capacity region.
\end{proof}

In the proof of Theorem \ref{theorem:scheduling:optimal}, the number of time slots in a frame, $M$, can be large, which means the truncated time debt of each client is updated infrequently. While the EPDF policy guarantees that the long-term average throughput of each client is at least as large as its requirement, a large $M$ may lead to undesirable short-term performance. The following example depicts one such scenario:

\begin{example}
Consider a system with two streams. Each stream generates one packet in each time slot, and requires a delay bound of $\tau_n=1$. Further, assume that $w_1=0.5$ and $w_2=0$, that is, client 2 does not require any of its packets to be delivered. Suppose that the EPDF policy chooses frame size $M=100$. Now, at the beginning of the first time slot, we have $d_1(1)=Mw_1=50$ and $d_2(1)=0$. By the EPDF policy, the AP schedules transmissions for client 1 for all time slots in $[1,50]$. After time slot 50, the truncated time debt for client 1 becomes 0, and the AP may schedule transmissions for client 2 for all time slots in $[51,100]$. To summarize, the EPDF may devote the first half of the frame to client 1, and then totally neglects it in the second half of the frame, which can result in poor short-term performance for client 1. On the other hand, if the EPDF policy chooses $M=2$, it will schedule client 1 in time slots $\{1,3,5,\dots\}$, and the short-term performance of client 1 will be much better than it is when $M=100$. $\Box$
\end{example}

To improve short-term performance, we may want to set $M$ to be a small number. However, the following example suggests that the EPDF policy may fail to support a vector of throughputs that is strictly in the capacity region when $M$ is too small.

\begin{example}	\label{example:smallM}
Consider a system with three clients. Client 1 generates a packet in every time slot, has channel reliability $p_1=1.0$, and requires that $\tau_1=1$, $q_1=0.5$, and hence $w_1=0.5$. Client 2 generates a packet in every time slot, has channel reliability $p_2=1.0$, and requires that $\tau_2=1$ and $q_2=0$. Finally, client 3 generates a packet in time slots of the form $4m+2$, $m=0,1,2,\dots$, i.e., in time slots $\{2,6,10,14,\dots\}$. Client 3 has channel reliability $p_3=0.5$, and requires that $\tau_3=2$, $q_3=3/16$, and hence $w_3=3/8=0.375$.

Suppose we set $M=4$. Then the EPDF policy schedules client 1 in time slots of the form $4m+1$ and $4m+2$, schedules client 3 in time slots of the form $4m+3$, and schedules client 3 in time slots of the form $4m+4$ if the transmission in time slot $4m+3$ fails, for all $m=0,1,2,\dots$. It is easy to check that the EPDF policy supports the required throughput to every client.

On the other hand, suppose we set $M=2$. At the beginning of time slot 1, we have $d_1(1)=1, d_2(1)=0, d_3(1)=0.75$. The EPDF policy may schedule client 1 in time slot 1, and client 2 in time slot 2. Note that the packet of client 3 has not been generated in time slot 2, and hence cannot be scheduled in the time slot. At the beginning of time slot 3, we have $d_1(3)=1-1+1=1$, $d_2(3)=0$, and $d_3(3) = 0.75+0.75=1.5$. The EPDF policy then schedules client 1 in time slot 3, as its deadline is earlier than that of client 3, and schedules client 3 in time slot 4. The same schedule then repeats for all the following time slots. Hence, under the EPDF policy, client 3 is scheduled once every 4 time slots, and has a throughput of $1/8$, which is less than its requirement. $\Box$
\end{example}

In the above example, the EPDF policy fails to provision the required throughputs when the frame size $M$ is too small. The reason is that when $M=2$, the packet generations differ greatly from frame to frame. In particular, client 3 only generates packets in half of the frames. In order to improve short-term performance without losing long-term performance, we choose $M$ so that the set of packets generated in each frame is similar.

\section{Simulations on Three Movies}  \label{section:simulations}

We now present simulation results based on real traces of video streaming. We use the traces provided by the Video Trace Library of the Arizona State University \cite{SeRe12,traces2}. We conduct our simulations on three different HD movies, namely, ``Harry Potter," ``Finding Neverland," and ``Transporter 2."

We have implemented the EPDF policy, as well as three other policies for performance comparison, in ns-2. The three other policies are the EDF policy, which is proposed in \cite{CLL73} and described in Section \ref{section:feasibility}, the Largest Debt First (LDF) policy, and a policy from Dua and Bambos \cite{AD07}. The LDF policy schedules transmissions for the client that has the largest truncated time debt in each time slot. It has been shown that the LDF policy achieves the capacity region for the special case when all streams generate packets at the same time, and none of the clients buffer incoming packets before playing the video \cite{IH09}\cite{IH09Hoc}. In \cite{AD07}, it is assumed that each packet drop due to expiration incurs a \emph{dropping cost}, where the dropping costs may be different from client to client. A scheduling policy is then proposed aiming to minimize the long-term average total dropping costs over all clients. This policy considers the dropping cost and deadline of packets, as well as the channel reliabilities of clients. In our simulations, we set the dropping cost of client $n$ to be $q_n/r_n$, as this is the proportion of packets that client $n$ requires to be delivered on time.

We consider IEEE 802.11a, which supports up to 54Mbps data rate, as the MAC protocol. Simulations show that the time needed to transmit a packet of size 1500 Bytes, including all overheads and the time for ACK, is about 750 $\mu s$. Hence, we set the length of a time slot to be 750 $\mu s$. When generating traffic from the traces, we divide a video frame into multiple packets if its size is larger than 1500 Bytes, and we also combine adjacent video frames into one packet when their total size is less than 1500 Bytes.

We consider a system with 30 clients arranged in a $6\times 5$ grid, where adjacent clients are separated by 170 meters, and one AP that is placed at the center of the grid. We use the \texttt{Shadowing} module in ns-2 to model the unreliable wireless links. Simulations show that under the above topology, the values of $p_n$ range from $51\%$ to $100\%$. In each simulation, the AP keeps a running estimate of the channel reliability for each client. When the AP transmits a packet for client $n$, it updates its estimate of $p_n$ by $p_n\leftarrow (1-\alpha)p_n+ \alpha\cdot 1(\mbox{the transmission succeeds})$, where $1(\cdot)$ is the indicator function.

Each simulation lasts 100 seconds. In each simulation, we assume that all clients are watching the same movie, but their starting times are staggered. To be more specific, we assume that when client $n$ is watching the $t$-th second of the movie, client $(n+1)$ is watching the $(t+100.11)$-th second of the same movie. This ensures that the traffic patterns are different from client to client. Unless otherwise specified, the delay bounds of clients, $\tau_n$, are assumed to be evenly distributed between 20000 -- 30000 time slots. 

To better illustrate the simulation results, we assume that half of the clients require a portion $X$ of their packets to be delivered on time, while the other half of the clients require a portion $Y$ of their packets to be delivered on time. The throughput requirements of clients are then computed accordingly. We define the \emph{achieved throughput region} of a policy as the region composed of all $(X,Y)$ that can be achieved by the policy. A pair $(X,Y)$ is considered to be achieved if the resulting throughput of each client is at least 95$\%$ of its required throughput.

\subsection{Performance Comparisons Between Different Policies}	\label{section:simulation:policy}

\begin{figure*}
\subfigure[Harry Potter]{
\label{fig:HarryPotter} %% label for first subfigure
\includegraphics[width=2in]{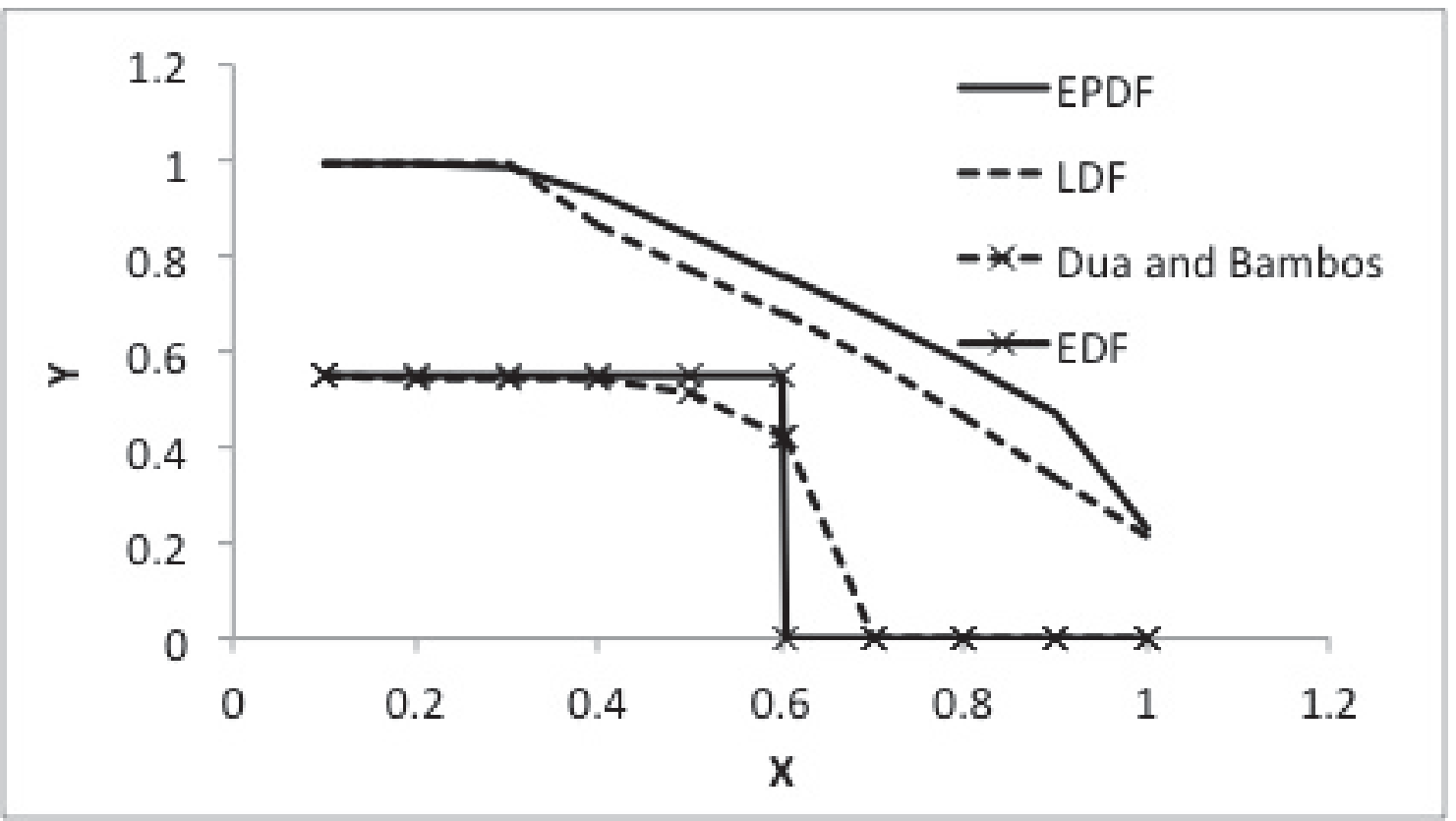}}
\hspace{0.01\linewidth} 
\subfigure[Finding Neverland]{
\label{fig:FindingNeverland} %% label for second subfigure
\includegraphics[width=2in]{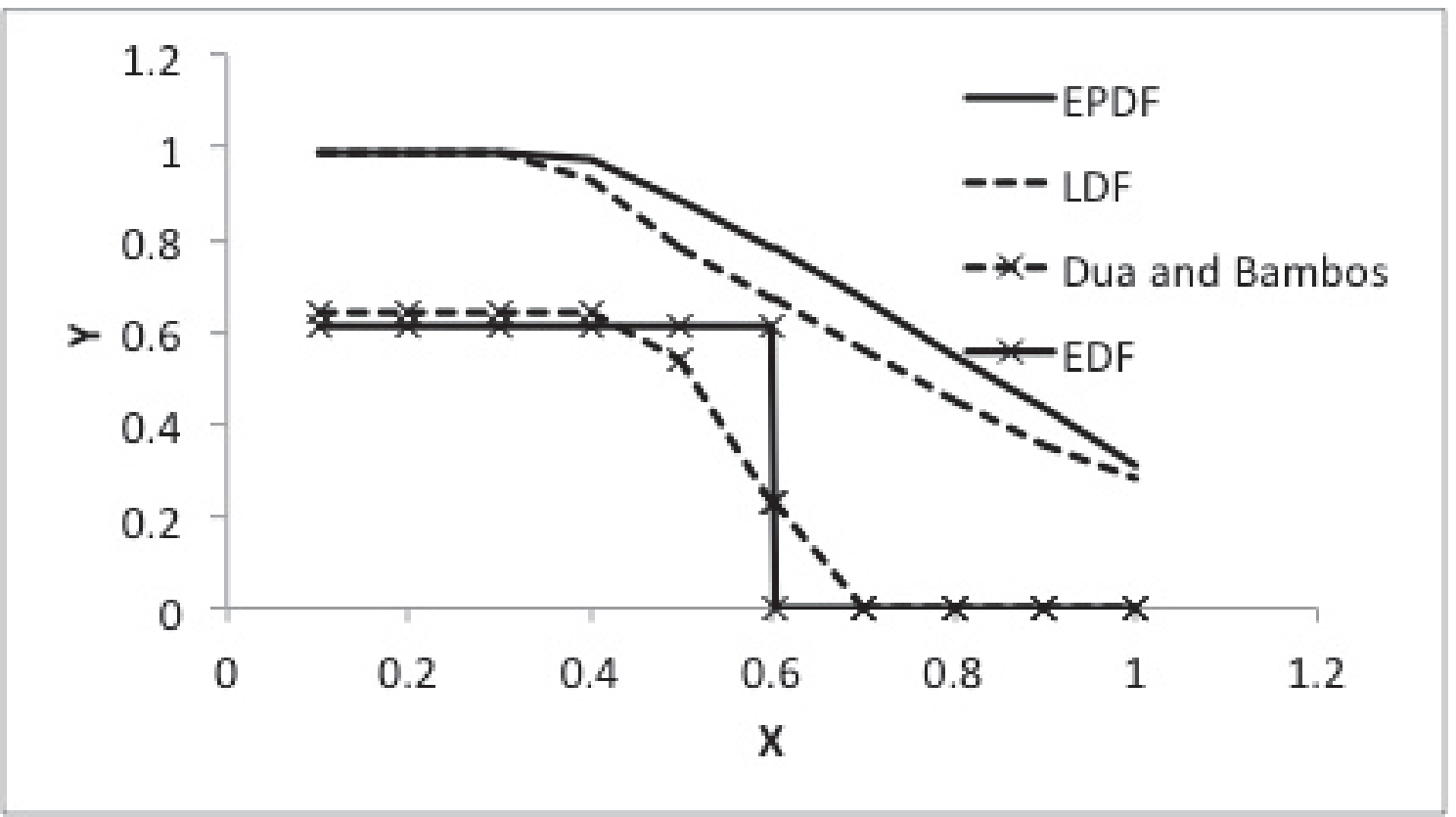}}
\hspace{0.01\linewidth} 
\subfigure[Transporter 2]{
\label{fig:Transporter2} %% label for first subfigure
\includegraphics[width=2in]{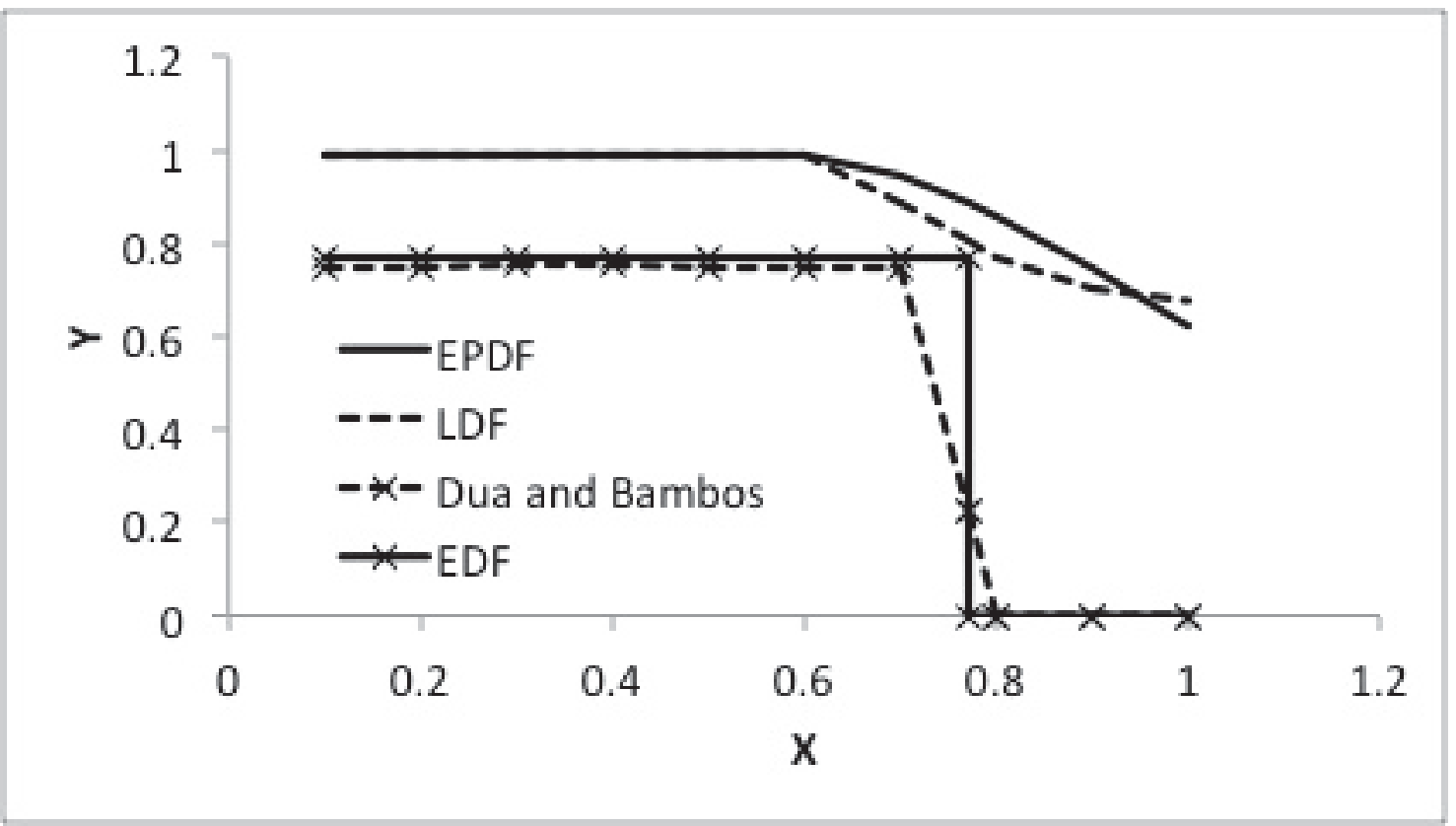}}
\caption{Achieved throughput regions under different policies}\label{fig:simulation:policy}
\end{figure*}

We first present the performance of all the four scheduling policies. Simulation results are shown in Fig. \ref{fig:simulation:policy}. It can be seen that the EPDF policy outperforms all the other three policies in all three movies. For a fixed $X$, the difference of achieved $Y$ between the EPDF policy and the LDF policy can be as large as 0.13. Since the LDF policy does not take into account the possibly different delay bounds for different packets, it may spend a lot of time transmitting packets whose debts are large but deadlines are faraway. Both the EDF policy and the policy in \cite{AD07} result in poor performance. The EDF policy does not take the throughput requirements of clients into account, and hence may be severely unfair. On the other hand, while the policy in \cite{AD07} makes scheduling decisions based on the deadlines of packets, the throughput requirements and channel reliabilities of clients, it still results in suboptimal performance. This suggests that only aiming to minimize the total dropping cost may be unfair to some flows, and hence the policy in \cite{AD07} is not suitable for providing performance guarantees for individual clients. 

\subsection{Impact of Delay Bounds}

\begin{figure*}[t]
\subfigure[Harry Potter]{
\label{fig:HarryPotter} %% label for first subfigure
\includegraphics[width=2in]{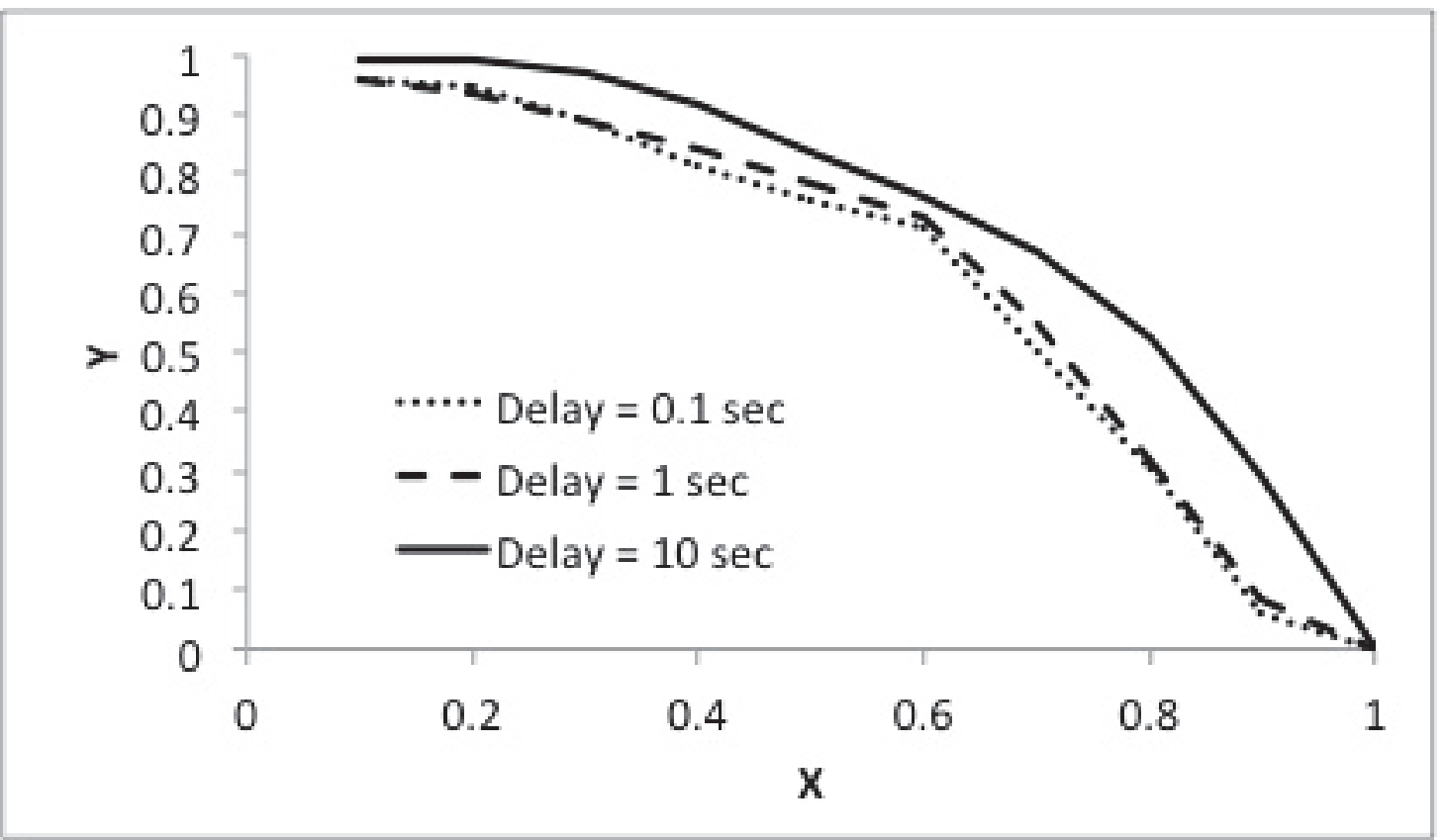}}
\hspace{0.01\linewidth} 
\subfigure[Finding Neverland]{
\label{fig:FindingNeverland} %% label for second subfigure
\includegraphics[width=2in]{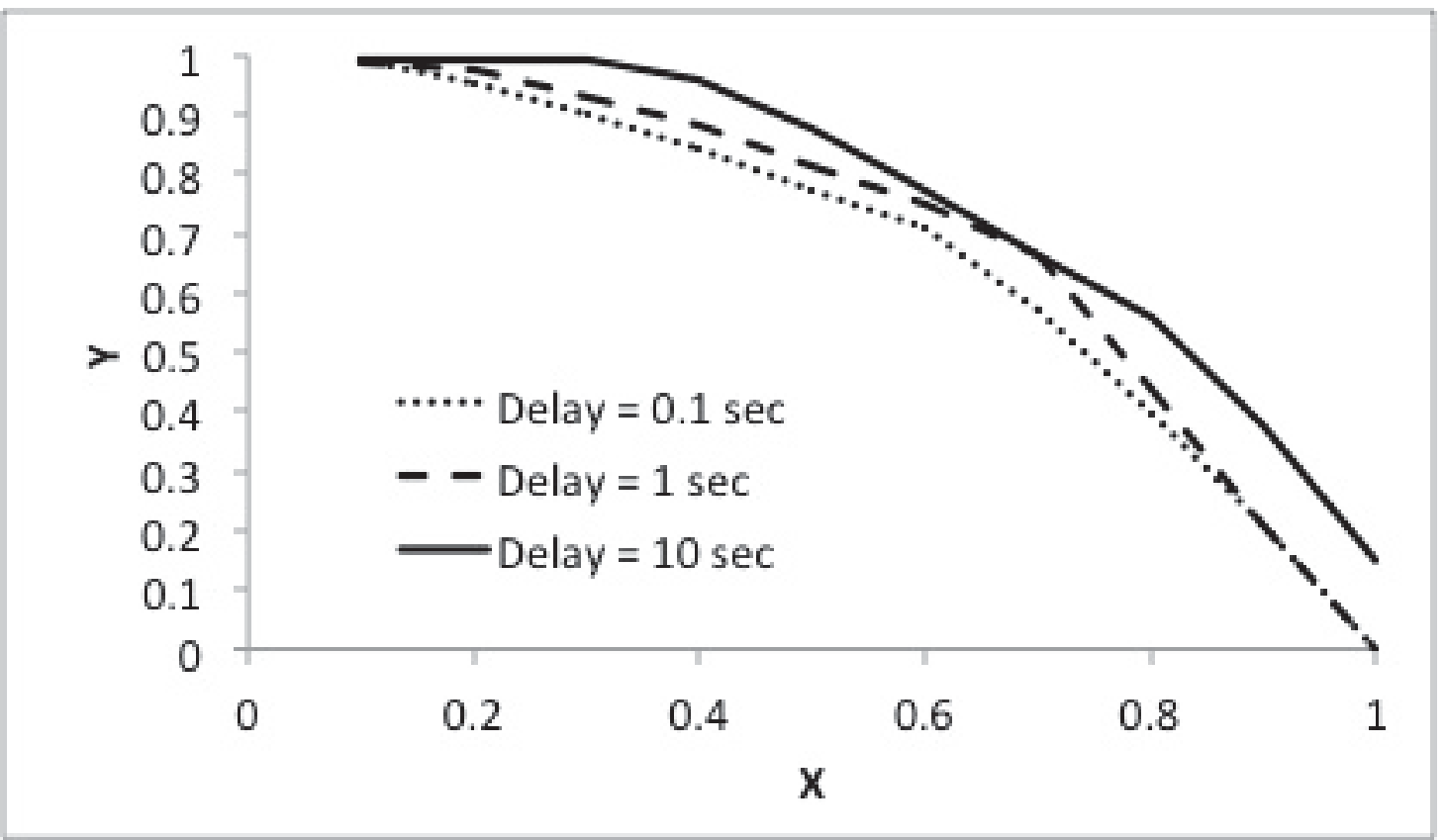}}
\hspace{0.01\linewidth} 
\subfigure[Transporter 2]{
\label{fig:Transporter2} %% label for first subfigure
\includegraphics[width=2in]{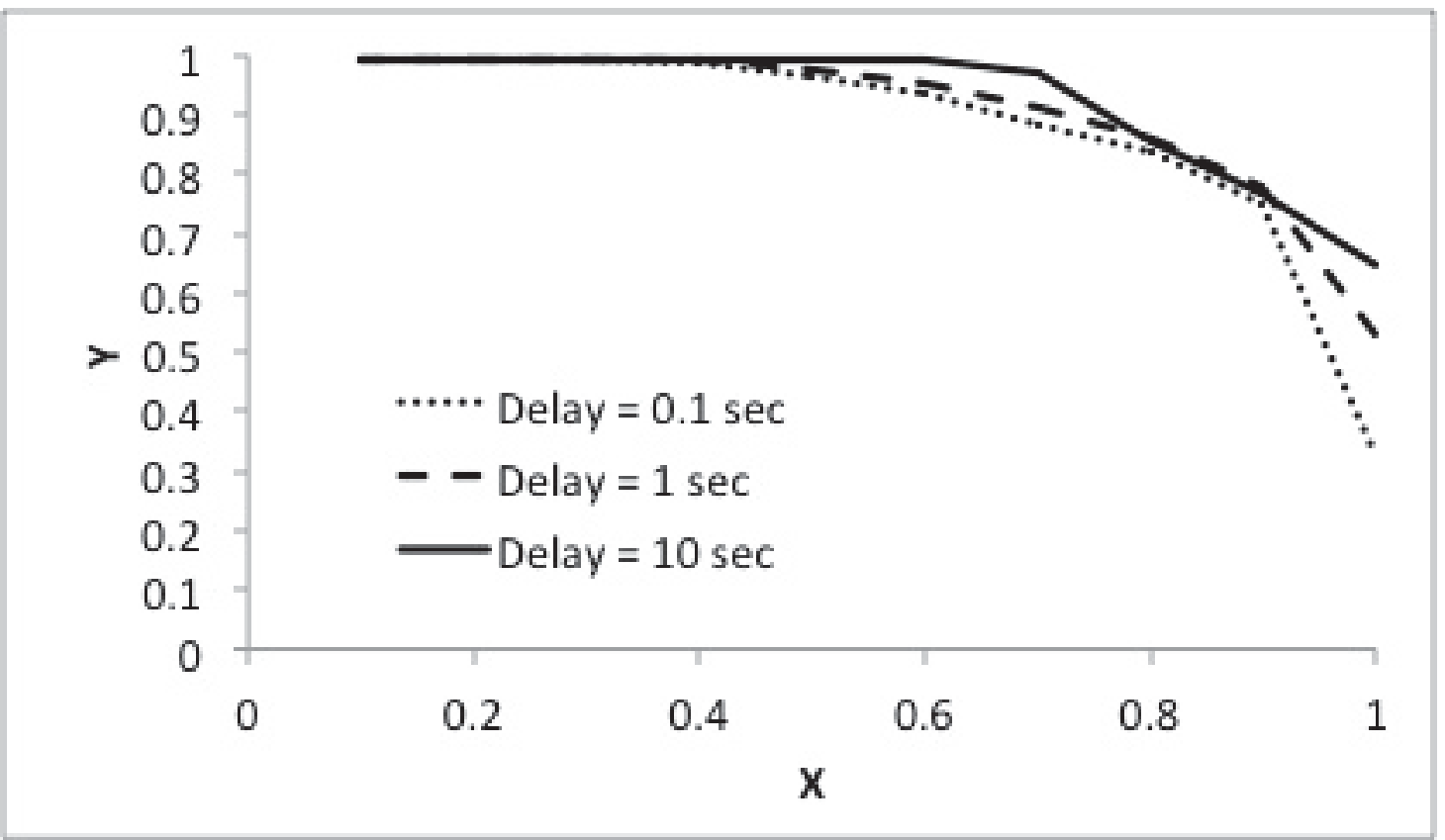}}
\caption{Achieved throughput regions for different delay bounds.}\label{fig:simulation:buffer}
\end{figure*}

Next, we investigate the influence of delay bounds on throughput performance. We again consider the system with 30 clients. The settings of the system are almost the same as in Section \ref{section:simulation:policy}. The only difference is that all clients require the same delay bound. We compare the achieved throughput regions for different delay bounds.

Simulation results are shown in Fig. \ref{fig:simulation:buffer}. In all three movies, the achieved throughput regions increase with delay bounds. For a fixed $X$, the largest achievable $Y$ can be increased by 0.2 when the delay bounds are increased from 0.1 second to 1 second. When we further increase the delay bound from 1 second to 10 seconds, the largest achievable $Y$ can be increased by 0.21 for some fixed $X$. As the delay bounds increase, live video streams become less vulnerable to bursty packets and bursty failed transmissions, and hence the throughput performance improves.

\subsection{Influence of Frame Sizes}

\begin{figure*}[t]
\subfigure[Harry Potter]{
\label{fig:HarryPotter} %% label for first subfigure
\includegraphics[width=2in]{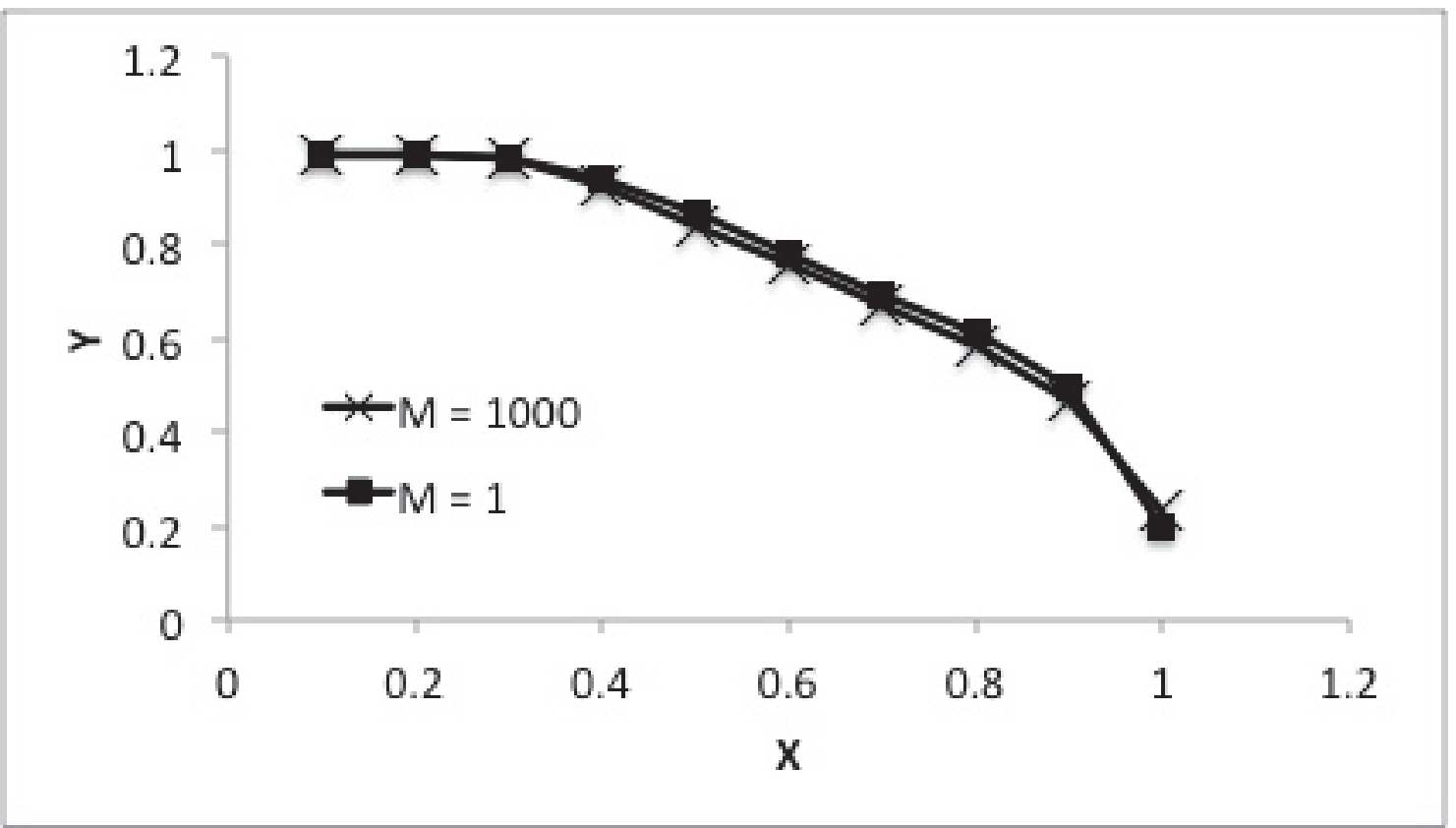}}
\hspace{0.01\linewidth} 
\subfigure[Finding Neverland]{
\label{fig:FindingNeverland} %% label for second subfigure
\includegraphics[width=2in]{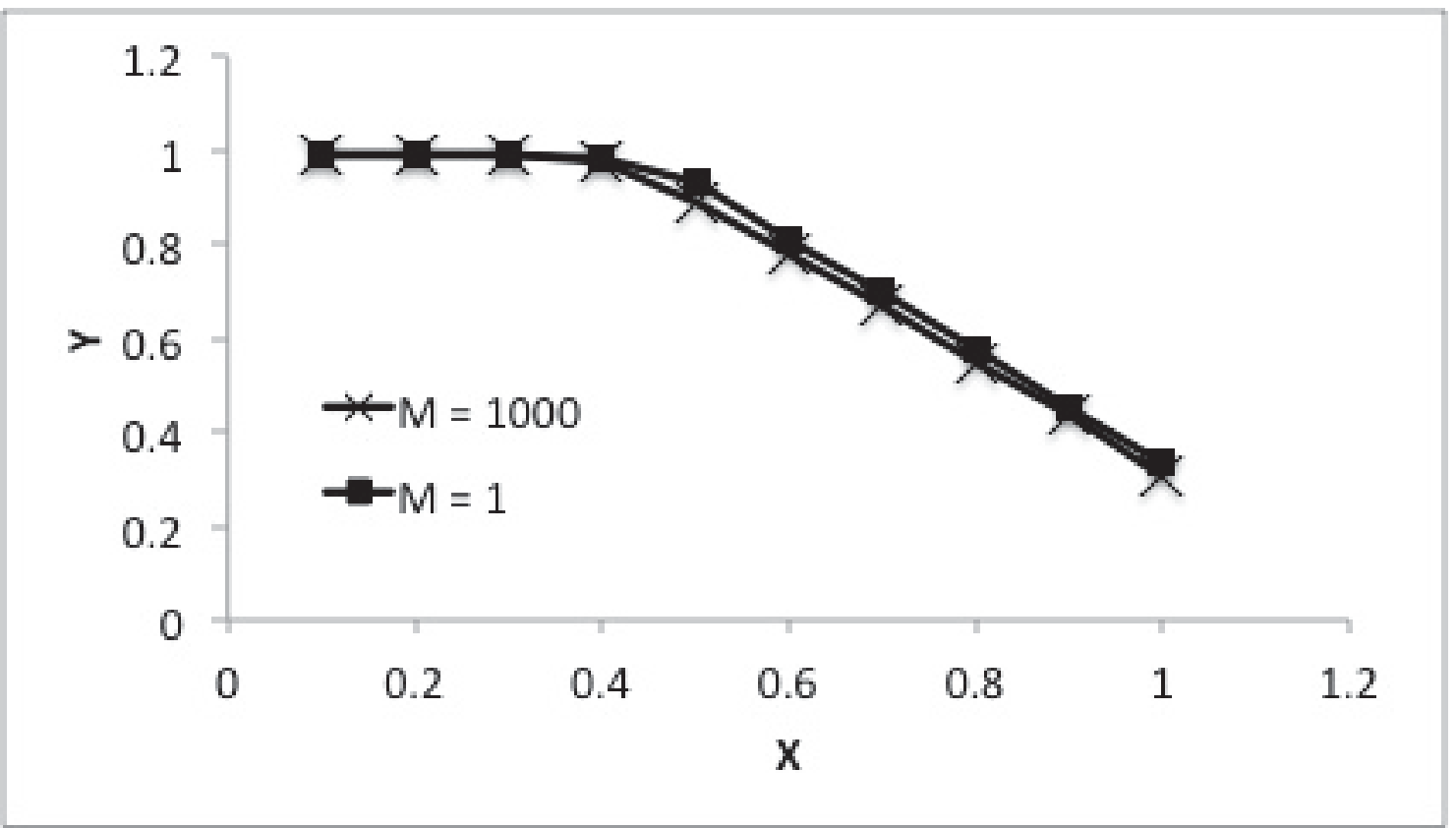}}
\hspace{0.01\linewidth} 
\subfigure[Transporter 2]{
\label{fig:Transporter2} %% label for first subfigure
\includegraphics[width=2in]{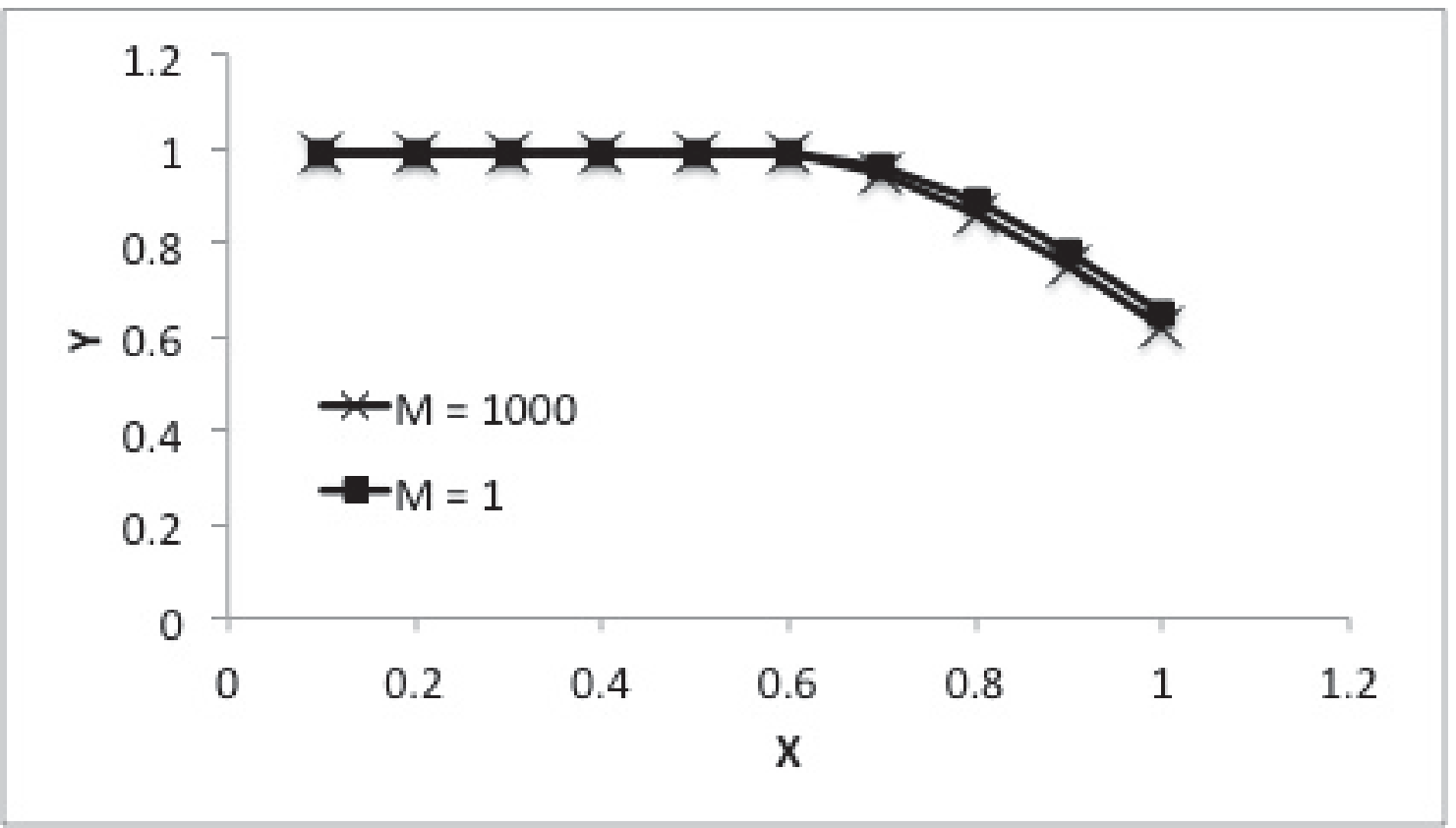}}
\caption{Achieved throughput regions for different frame sizes.}\label{fig:simulation:M}
\end{figure*}

In Example \ref{example:smallM}, it is shown that the EPDF policy may not be capacity achieving if the frame size, i.e., the value of $M$, is too small. In this section, we investigate the influence of frame sizes on system performance.

We consider the same system as that in Section \ref{section:simulation:policy}. We simulate the performance of the EPDF policy under two extreme cases: $M=1$ and $M=1000$. Simulation results are shown in Fig. \ref{fig:simulation:M}. Surprisingly, the achieved throughput regions for the two cases are virtually indistinguishable. For any movie and any fixed $X$, the difference between the achievable $Y$ is less than 0.03 under the two cases.

Next, we investigate the short-term performance under both cases of $M=1$ and $M=1000$. We assume that all clients require a portion $X$ of their packets to be delivered on time. The value of $X$ is set to be 0.65 for Harry Potter, 0.7 for Finding Neverland, and 0.8 for Transporter 2, since Fig. \ref{fig:simulation:M} shows that these settings are feasible. For each movie, we find the client that generates the most packets during the simulation, and then plot the number of packets that this client receives in every second of the simulations.

Simulation results are shown in Fig. \ref{fig:simulation:short}. The number of packets delivered to the client in each second varies much from second to second when $M=1000$. On the other hand, when $M=1$, the number of packets delivered to the client does not vary as much. Although both $M=1000$ and $M=1$ result in similar long-term throughputs for each client, smaller $M$ can lead to much better short-term performance.

\begin{figure*}[t]
\subfigure[Harry Potter]{
\label{fig:HarryPotter} %% label for first subfigure
\includegraphics[width=2in]{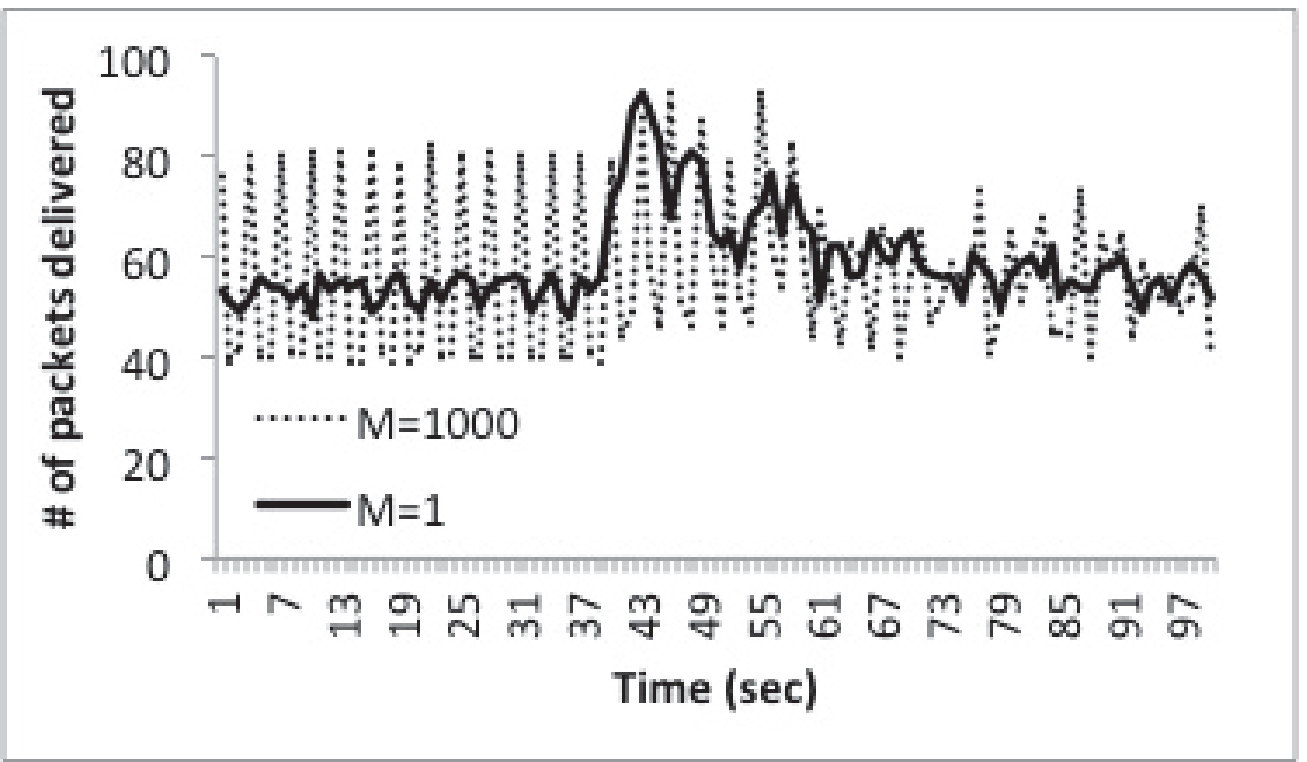}}
\hspace{0.01\linewidth} 
\subfigure[Finding Neverland]{
\label{fig:FindingNeverland} %% label for second subfigure
\includegraphics[width=2in]{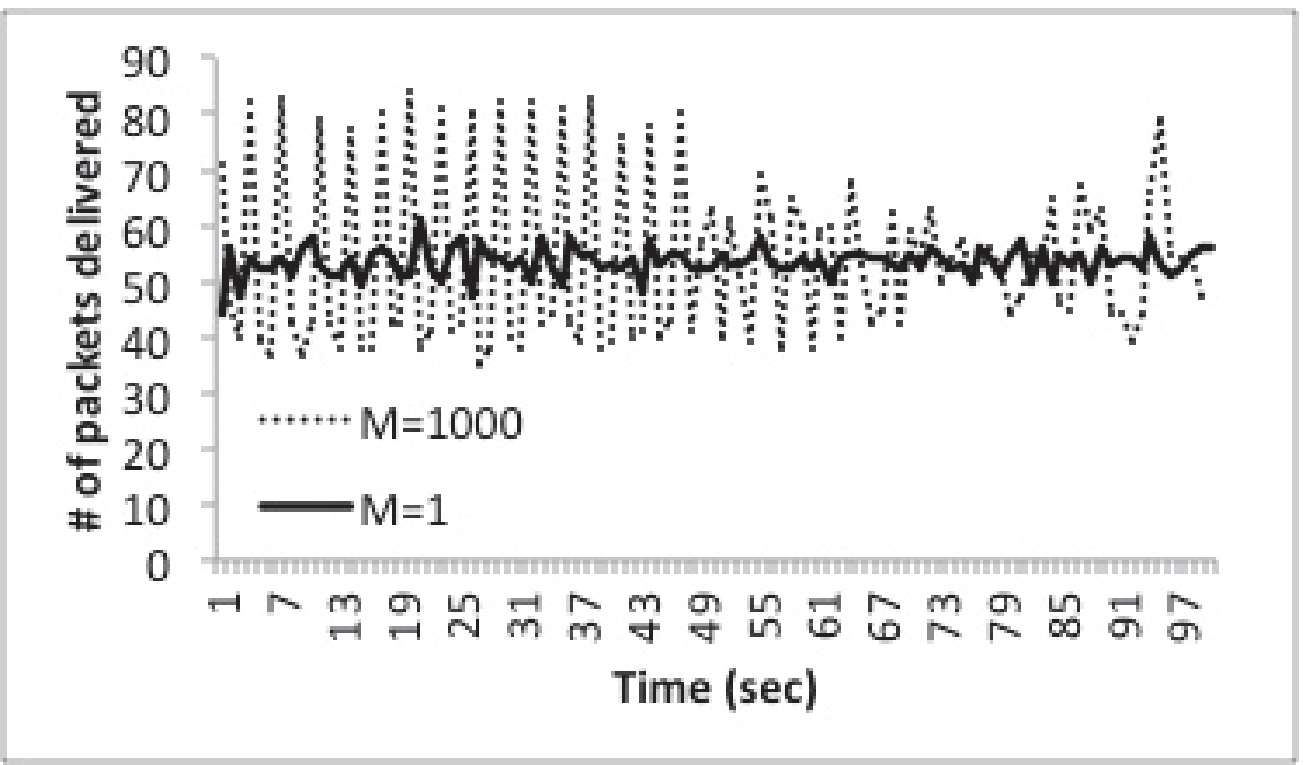}}
\hspace{0.01\linewidth} 
\subfigure[Transporter 2]{
\label{fig:Transporter2} %% label for first subfigure
\includegraphics[width=2in]{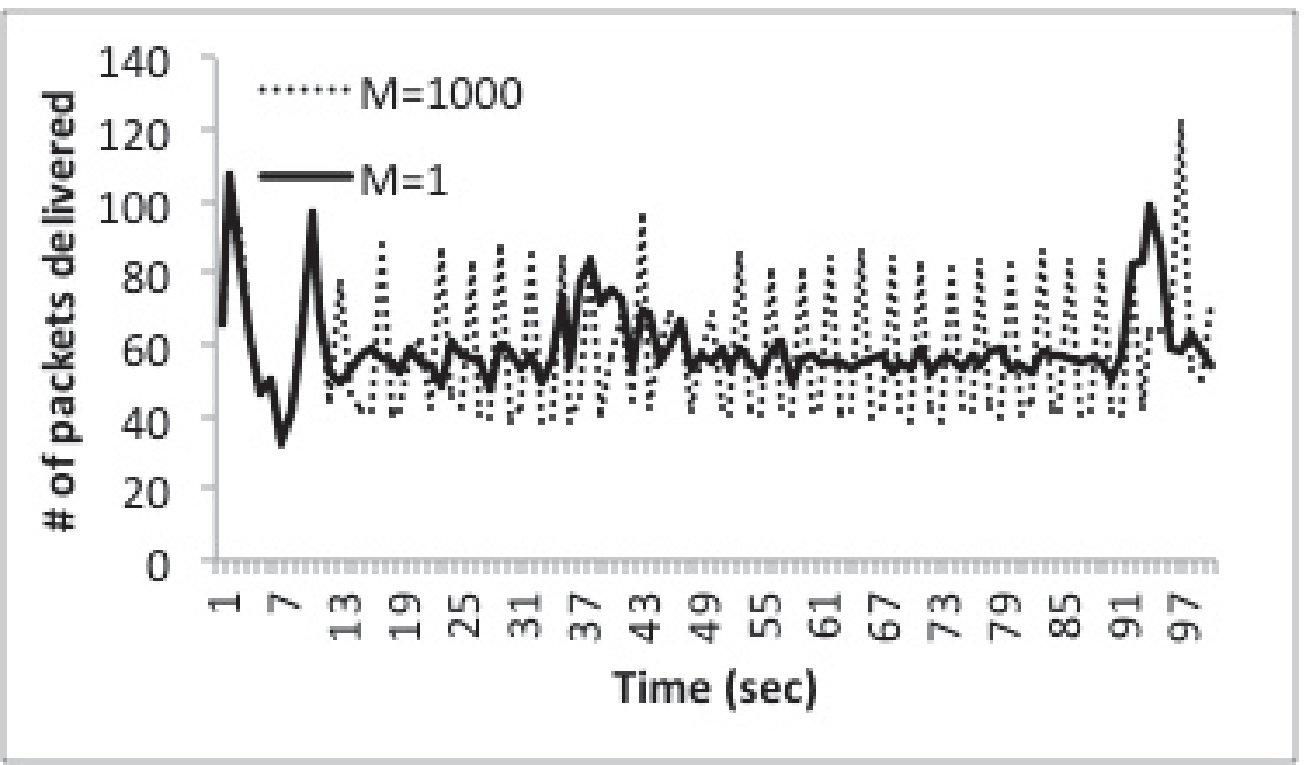}}
\caption{Achieved throughput regions for different frame sizes.}\label{fig:simulation:short}
\end{figure*}

\section{Conclusions}	\label{section:conclusions}

We have studied the problem of serving multiple live video streams to wireless clients with various delay and throughput requirements, as well as heterogeneous link reliabilities. We have analyzed a model that jointly considers several practical concerns for live video streaming over wireless networks, including the traffic patterns of streams, the unreliabilities of wireless transmissions, and the delay and throughput requirements of clients. We have presented a systematic approach for precisely characterizing the capacity region. We have proposed a simple scheduling policy, the EPDF policy, and proved that it is able to support every vector of throughputs that is strictly in the capacity region. The utility of the EPDF policy has been demonstrated through trace-based simulations, where it is seen that it appears to outperform other policies by a large margin.

%\def\baselinestretch{0.7}
%\small
\bibliographystyle{acm}
\bibliography{reference}

%\def\baselinestretch{0.5}
%\normalsize
\end{document}